\documentclass[reprint, superscriptaddress, amsmath, amssymb, aps, prl,longbibliography]{revtex4-2}
\usepackage{color}
\usepackage{graphicx}
\usepackage{dcolumn}
\usepackage{bm}
\usepackage{braket}
\usepackage{enumerate}
\usepackage{amsmath,amssymb,amsfonts,ascmac,mathtools}
\usepackage{MnSymbol}

\usepackage{amsthm}
\theoremstyle{definition}
\newtheorem{thm}{Theorem}
\newtheorem{prop}[thm]{Proposition}
\newtheorem{lem}[thm]{Lemma}

\renewcommand\paragraph[1]{%
  \par\emph{#1---}\kern2pt\relax\ignorespaces}

\usepackage{hyperref}
\hypersetup{
  colorlinks=true,
  citecolor=magenta,
  linkcolor=blue,
  urlcolor=cyan,
  breaklinks=true
}

\DeclarePairedDelimiter{\abs}{\lvert}{\rvert}

\renewcommand{\ket}[1]{\mathinner{|{#1})}}

\newcommand{\NSstates}{\mathcal{P}_{\mathrm{NS}}}
\newcommand{\Ent}{\mathcal{V}_\mathrm{PR}}
\newcommand{\Sep}{\mathcal{V}_\mathrm{LR}}
\newcommand{\Observables}{\mathcal{O}}
\newcommand{\ext}{\mathrm{ext}}

\makeatletter
\AtBeginDocument{%
  \let\thebibliography\rtx@thebibliography
  \let\endthebibliography\endrtx@thebibliography
}
\makeatother
\makeatletter

\def\REV@fm@prefix{frontmatter.}
\def\REV@fm@dest#1{\REV@fm@prefix#1}

\def\frontmatter@footnotemark#1{%
  \leavevmode
  \ifhmode\edef\@x@sf{\the\spacefactor}\nobreak\fi
  \begingroup
    \hyper@linkstart{link}{\REV@fm@dest{#1}}%
      \csname c@\@mpfn\endcsname#1\relax
      \def\@thefnmark{\frontmatter@thefootnote}%
      \@makefnmark
    \hyper@linkend
  \endgroup
  \ifhmode\spacefactor\@x@sf\fi
  \relax
}

\def\present@bibnote#1#2{%
  \item[%
    \textsuperscript{%
      \normalfont
      \Hy@raisedlink{%
        \hyper@anchorstart{\REV@fm@dest{#1}}\hyper@anchorend
      }%
      \begingroup
        \csname c@\@mpfn\endcsname#1\relax
        \frontmatter@thefootnote
      \endgroup
    }%
  ]#2\par
}

\def\present@FM@footnote#1#2{%
  \begingroup
    \csname c@\@mpfn\endcsname#1\relax
    \def\@thefnmark{\frontmatter@thefootnote}%
    \Hy@raisedlink{%
      \hyper@anchorstart{\REV@fm@dest{#1}}\hyper@anchorend
    }%
    \frontmatter@footnotetext{#2}%
  \endgroup
}

\newcounter{REV@maketitle@count}

\def\frontmatter@maketitle{%
  \stepcounter{REV@maketitle@count}%
  \begingroup
    \ifnum\value{REV@maketitle@count}>1\relax
      \def\REV@fm@prefix{SMfrontmatter.}%
      \let\frontmatter@footnote@produce
          \frontmatter@footnote@produce@footnote
    \fi
    \let\REV@orig@label\label
    \def\label##1{%
      \edef\REV@tmp{##1}%
      \edef\REV@tgt{FirstPage}%
      \ifx\REV@tmp\REV@tgt
      \else
        \REV@orig@label{##1}%
      \fi
    }%
    \@author@finish
    \title@column\titleblock@produce
    \suppressfloats[t]%
    \let\abstract\@undefined
    \let\endabstract\@undefined
    \titlepage@sw{\vfil\clearpage}{}%
    \onecolumn@grid@setup
    \def\set@footnotewidth{\set@footnotewidth@one}%
  \endgroup
}

\makeatother

\makeatletter
\newcommand{\equalcontrib}{%
  \frontmatter@footnote{S.~U.~and A.~H. contributed equally to this work.}%
}
\makeatother

\begin{document}
\newcommand{\nameoftitle}{Entanglement Generation Beyond Quantum Theory: \\ 
From Product States to Popescu–Rohrlich Boxes}
\title{\nameoftitle}

\author{Shun Umekawa\equalcontrib}
\email{umeshun2003@g.ecc.u-tokyo.ac.jp}
\affiliation{Department of Physics, The University of Tokyo, 7-3-1 Hongo, Tokyo 113-0033, Japan}

\author{Akihiro Hokkyo\equalcontrib}
\email{hokkyo@cat.phys.s.u-tokyo.ac.jp}
\affiliation{Department of Physics, The University of Tokyo, 7-3-1 Hongo, Tokyo 113-0033, Japan}

\author{Hayato Arai}
\email{h.arai6626@gmail.com}
\affiliation{Communication Science Laboratories, NTT, Inc., 3-1 Morinosato Wakamiya, Atsugi, Kanagawa 243-0198, Japan}
\affiliation{Department of Basic Science, The University of Tokyo, 3-8-1 Komaba, Tokyo 153-8902, Japan}

\author{Kazuaki Takasan}
\email{takasan@ap.t.u-tokyo.ac.jp}
\affiliation{Department of Applied Physics, The University of Tokyo, 7-3-1 Hongo, Tokyo 113-8656, Japan}

\begin{abstract}
Entanglement generation is a fundamental dynamical capability in quantum information science and underpins many quantum advantages.
While quantum theory enables it through unitary dynamics, boxworld, a generalized probabilistic theory admitting Popescu--Rohrlich boxes with supraquantum correlations, 
has no reversible transformation capable of generating entanglement.
We show that this no-go picture changes fundamentally once reversibility is relaxed to pure-state preservation. 
We construct a pure-state-preserving transformation that maps 
every uncorrelated pure state
to a Popescu--Rohrlich box and completely classify all pure-state-preserving entangling
transformations in the simplest bipartite boxworld.
Our results provide the first explicit mechanism for generating beyond-quantum entanglement without introducing mixing and demonstrate
a physical distinction between reversibility and pure-state preservation that is obscured by the structure of quantum theory.

\end{abstract}

\maketitle

\paragraph{Introduction}
Entanglement is one of the central concepts of quantum theory \cite{Horodecki2009}.
In quantum information processing, it serves as an essential resource for a wide range of protocols, including quantum key distribution~\cite{Ekert1991} and quantum teleportation~\cite{Teleportation}.
From a foundational perspective, Bell's theorem~\cite{Bell} and, in particular, the Clauser--Horne--Shimony--Holt (CHSH) inequality~\cite{CHSH},
\begin{equation}
|\langle X X\rangle +a \langle X Y\rangle + b\langle Y X\rangle - ab\langle Y Y\rangle| \le 2,
\label{eq:CHSH}
\end{equation}
where \(a,b\in\{\pm1\}\),
quantitatively separate correlations compatible with local realism from
certain correlations attainable in quantum theory using entangled states, leading to experimental tests of the incompatibility between locality and realism~\cite{Aspect1982}.

Although entanglement enables correlations beyond local realism, quantum theory itself imposes a strict limit on their strength.
In quantum theory, the CHSH value, i.e., the left-hand side of Eq.~\eqref{eq:CHSH}, is bounded by 
\(2\sqrt{2}\)
~\cite{Tsirelson}.
By contrast, 
Khalfin and Tsirelson~\cite{Khalfin1985} and, independently, 
Rastall~\cite{Rastall1985} showed that the no-signaling principle, 
an operational consequence of relativistic causality,
permits probabilistic models attaining the algebraic maximum 
\(4\); 
these models were later highlighted by Popescu and Rohrlich~\cite{PRbox} and are now known as Popescu--Rohrlich (PR) boxes.
The empirical absence of such ``beyond-quantum'' correlations has motivated the search for additional physical principles that 
rule out such correlations, 
including information causality~\cite{InfoCausality}.

Most of these studies address a \emph{static} question:
how strong the correlations of a given state, or box, can be.
Here we instead focus on a \emph{dynamical} question in beyond-quantum theories:
can one generate nonlocal correlations, and hence entanglement, from an uncorrelated product state?
In quantum theory, when two qubits are initially in a product state, entanglement can be generated straightforwardly by a two-qubit unitary such as the CNOT gate.
In beyond-quantum theories, however, 
even the existence of nontrivial entangling dynamics remains unclear.
To the best of our knowledge, no explicit entangling transformation has been exhibited other than maps that merely prepare entangled output states, namely replacement maps and measure-and-prepare maps. 

\begin{figure}
    \centering
    \includegraphics[width=9cm]{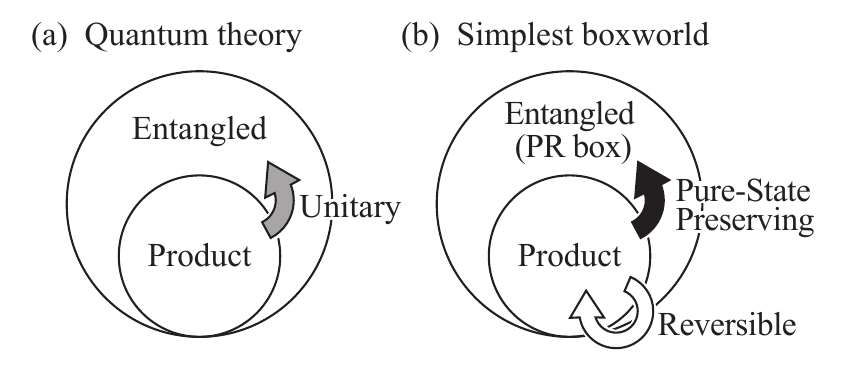}
    \caption{
    Entanglement generation in terms of two hallmark properties of unitary dynamics: reversibility and pure-state preservation. (a) In quantum theory, unitary dynamics can reversibly transform product pure states into entangled pure states. (b) Although the simplest boxworld, which corresponds to the setting for the CHSH inequality, admits no reversible entangling transformation, relaxing reversibility while retaining pure-state preservation allows product pure states to be transformed into entangled PR boxes.    
    }
    \label{fig:concept}
\end{figure}

Indeed, within the class of reversible transformations,
Gross \emph{et al.}~\cite{Gross2010} showed that
entanglement generation is impossible in boxworld~\cite{Barrett2007}, 
a generalized probabilistic theory 
that admits PR boxes,
concluding that ``there is no boxworld analogue of an entangling unitary''.
This provides a paradigmatic no-go result for reversible entanglement generation in a beyond-quantum theory, 
and it is conjectured that the same obstruction extends to maximally nonlocal theories more generally~\cite{Al-Safi2015}.
Thus, even when a theory admits correlations stronger than those of quantum theory, the dynamics capable of creating such correlations remain largely unexplored.

In this Letter, we show that this picture changes drastically once the assumption of reversibility is relaxed,
and explicitly construct a nontrivial \emph{pure-state-preserving} (PSP) transformation that generates a PR box from an uncorrelated product state (Fig.~\ref{fig:concept}).
PSP transformations map pure states to pure states, thereby capturing yet another hallmark of unitary dynamics distinct from reversibility.
In quantum theory, apart from trivial exceptions, the PSP property is essentially equivalent to reversibility~\cite{Hou2013}; 
in boxworld, however, our result shows that the two notions can be strictly separated from the perspective of entanglement generation.
Furthermore, we completely classify all such PSP entangling transformations 
in the simplest boxworld, showing that no examples exist beyond the family constructed here.

Our result sharpens the scope of the no-go theorem of Gross \emph{et al.}~\cite{Gross2010}: a distinctive feature of quantum theory may lie not in entanglement generation \textit{per se}, but in the ability to generate entanglement reversibly.
More broadly, our findings show that what is often subsumed under ``unitarity'' comprises at least two physically distinct principles—reversibility and pure-state preservation—which can come apart operationally beyond quantum theory.

\paragraph{No-signaling state space and PR boxes}
We consider the standard CHSH scenario.  Alice and Bob each choose one of two possible measurements, denoted by \(X\) and \(Y\), and each measurement has outcomes in \(\{\pm1\}\).  
We write \(\vec Z=(Z_A,Z_B)\in\{X,Y\}^2\) for the pair of measurement choices and
\(\vec m=(m_A,m_B)\in\{\pm1\}^2\) for the pair of outcomes. 
A fundamental object is a family of conditional probability distributions \(p=\{p(\vec m\mid \vec Z)\}_{\vec{m},\vec{Z}}\) .
This family is called a \textit{state}~\cite{Barrett2007,Gross2010}, 
because it completely determines the outcome probabilities
for all protocols involving local $X$ and $Y$ measurements, including one-way adaptive protocols.

The no-signaling condition~\cite{Eberhard1978,Rastall1985,PRbox,Barrett2005}
requires that each party's marginal distribution be independent of the other party's measurement choice:
\begin{equation}
\label{eq:no-signaling}
\begin{aligned}
\sum_{m_B=\pm1} p(m_A,m_B\mid Z_A,X)
&=
\sum_{m_B=\pm1} p(m_A,m_B\mid Z_A,Y),
\\
\sum_{m_A=\pm1} p(m_A,m_B\mid X,Z_B)
&=
\sum_{m_A=\pm1} p(m_A,m_B\mid Y,Z_B),
\end{aligned}
\end{equation}
which excludes superluminal communication. 
The state space is defined by the convex set of normalized nonnegative distributions satisfying Eq.~\eqref{eq:no-signaling}, which is denoted by \(\NSstates\).
This also corresponds to the simplest bipartite case of boxworld~\cite{Barrett2007}.
For $p\in\NSstates$, the marginal distributions $p_A(m_A\mid Z_A)$ and $p_B(m_B\mid Z_B)$
are defined by the expressions in the first and second lines of Eq.~\eqref{eq:no-signaling}, respectively.

For a no-signaling distribution \(p\), we define the expectation values as 
\begin{align}
A_Z(p)&\coloneqq\sum_{m_A\in\{\pm1\}}m_A\,p_A(m_A|Z), \\
B_Z(p)&\coloneqq\sum_{m_B\in\{\pm1\}}m_B\,p_B(m_B|Z), \\
C_{\vec Z}(p)&\coloneqq\langle Z_AZ_B\rangle_p
=\sum_{\vec m\in\{\pm1\}^2}m_Am_B\,p(\vec m|\vec Z),
\end{align}
where \(Z\in\{X,Y\}\) and \(\vec Z=(Z_A,Z_B)\in\{X,Y\}^2\).  
Due to the no-signaling condition, every $p\in\NSstates$ admits the following expression:
\begin{equation}\label{eq:expectation-coordinates}
p(\vec m|\vec Z)
=\frac14\{1+m_AA_{Z_A}(p)+m_BB_{Z_B}(p)+m_Am_BC_{\vec Z}(p)\}.
\end{equation}
Conversely, every set of eight numbers
\(\{A_Z,B_Z\}_{Z\in\{X,Y\}}\), 
\(\{C_{\vec{Z}}\}_{\vec{Z}\in\{X,Y\}^2}\) satisfying 
\begin{equation}\label{eq:positivity}
    \abs{A_{Z_A}-\delta B_{Z_B}}\le 1- \delta C_{\vec{Z}},\ \delta=\pm1
\end{equation}
determines $p\in\NSstates$ through Eq.~\eqref{eq:expectation-coordinates}. 
Therefore, 
the state space
can be identified with 
such a set of eight coordinates~\cite{Brunner2014}.

For \(a,b\in\{\pm1\}\), the CHSH values for \(p\) are defined as
\begin{equation}
S_{a,b}(p)=\braket{XX}_p+a\braket{XY}_p+b\braket{YX}_p-ab\braket{YY}_p.
\end{equation}
The CHSH inequalities~\cite{CHSH} state that \(\abs{S_{a,b}(p)}\le 2\) for all \(a,b\) if (and only if) \(p\) admits a local realistic model \cite{Fine1982}, 
in the sense that it can be written as a convex combination of product distributions of the form \(p_A(m_A\mid Z_A)\,p_B(m_B\mid Z_B)\).
Entangled quantum states can exhibit 
correlations that surpass the CHSH inequalities, 
but they obey \(\abs{S_{a,b}(p)}\le 2\sqrt2\)~\cite{Tsirelson}.  By contrast, no-signaling states can attain the algebraic maximum \(4\)~\cite{Khalfin1985,Rastall1985,PRbox}.
We define \(p_{\vec{n}}\in\NSstates\) as
\begin{equation}
p_{\vec{n}}(\vec{m}\mid \vec{Z})=\frac{1}{2}\delta_{m_Am_B,n_{\vec{Z}}},
\end{equation}
where \(\vec{n}=(n_{\vec{Z}})_{\vec{Z}\in \{X,Y\}^2}\) with \(n_{\vec{Z}}\in \{\pm1\}\).
This condition is equivalent to $A_Z=B_Z=0,C_{\vec{Z}}=n_{\vec{Z}}$. 
Thus, for \(\vec{n}\) satisfying \(\prod_{\vec{Z}}n_{\vec{Z}}=-1\), one can verify that the CHSH value attains \(|S_{a,b}(p_{\vec{n}})|=4\) for \(a=n_{X,X}/n_{X,Y},b=n_{X,X}/n_{Y,X}\in\{\pm1\}\).
Such \(p_{\vec{n}}\) are referred to as \textit{Popescu--Rohrlich (PR) boxes}~\cite{Barrett2005}.
This shows that the no-signaling condition alone does not single out the set of quantum correlations 
and allows nonlocal correlations stronger than those permitted by quantum theory.

\paragraph{Entanglement generation in the simplest bipartite boxworld}
In the following, we investigate \textit{entanglement generation} in the system 
with the state space \(\NSstates\).
A state is called entangled if it
does not admit any local realistic model.
We denote the set of states that admit a local realistic model as \(\mathcal{P}_{\mathrm{LR}}\),
and a state \(p\) is entangled if and only if \(p\in\NSstates\setminus\mathcal{P}_{\mathrm{LR}}\)
~\footnote{This is consistent with the terminology of ``entanglement'' in generalized probabilistic theories \cite{Barrett2007,Janotta2011,Janotta2012,Aubrun2021,Aubrun2022,Plavala2023}. \(\NSstates\) and \(\mathcal{P}_{\mathrm{LR}}\) 
correspond to the maximal and minimal tensor product state spaces of two square models (gbits), respectively.}.
We say that an allowed dynamical transformation generates entanglement, or is entangling, if it maps at least one non-entangled state to an entangled state.
Following Refs.~\cite{Gross2010,Al-Safi2015,Plavala2023}, 
we define allowed dynamical transformations as affine maps from the state space to itself~\footnote{This convention is a dynamical analogue of the so-called no-restriction hypothesis~\cite{Heinosaari2019}. 
More restrictive formulations instead treat the set of allowed transformations as an additional component of the theory, 
defining it as a subset of such maps~\cite{Hardy2001,Barrett2007}.}.
Thus, an affine map $\phi:\NSstates\to\NSstates$ is entangling if  \(\phi(\mathcal{P}_{\mathrm{LR}})\nsubseteq \mathcal{P}_{\mathrm{LR}}\).
Our main interest lies in nontrivial entanglement generation, 
beyond merely replacing the input state or, more generally, applying a measure-and-prepare map.

In quantum theory, interactions generate entanglement through unitary dynamics.
This motivates asking whether ``unitary-like'' dynamics on \(\NSstates\) can generate a PR box.
Previous works have shown no-go results: \textit{reversible} boxworld dynamics are nonentangling~\cite{Gross2010,Al-Safi2014,Al-Safi2015}.
We focus on yet another characterization of the unitary dynamics: \textit{pure-state preservation}.
Every reversible affine transformation is pure-state-preserving (PSP)~\cite{Hardy2001,Gross2010}, but the converse need not hold.
In quantum theory, however, these two properties are essentially equivalent; 
every completely positive trace-preserving map on a finite-dimensional quantum system that sends every pure state to a pure state is unitary, except for a constant map outputting a fixed pure state~\cite{Hou2013}.
We ask whether PSP dynamics can nevertheless generate entanglement in boxworld.

A state is pure if and only if it is an extreme point of the state space~\cite{Barrett2007}.
The pure states of \(\NSstates\) are the sixteen deterministic boxes and the eight PR boxes~\cite{Barrett2005}.
We denote these sets by \(\Sep\) and \(\Ent\), respectively, so that
\(
\ext(\NSstates)=\Sep\cup\Ent.
\)
Each \(p\in\Sep\) is characterized by
\(\vec n^A,\vec n^B\in\{\pm1\}^{\{X,Y\}}\) satisfying
\begin{equation}
A_Z=n^A_Z,\qquad
B_Z=n^B_Z,\qquad
C_{\vec Z}=n^A_{Z_A}n^B_{Z_B}.
\label{eq:characterization_sep_pure}
\end{equation}
We seek an entangling PSP transformation, equivalently an affine map \(\phi:\NSstates\to\NSstates\) satisfying
\begin{align}
\phi\bigl(\ext(\NSstates)\bigr)&\subseteq\ext(\NSstates),
\label{eq:PSP_definition}\\
\phi(\Sep)\cap\Ent&\neq\emptyset.
\label{eq:PSP_entangling_condition}
\end{align}
The equivalence follows from
\(\mathcal{P}_{\mathrm{LR}}=\operatorname{conv}(\Sep)\) and the affineness of \(\phi\).

\paragraph{Construction of entangling PSP maps}
We now construct PSP maps that generate PR-box entanglement.  
One can consider the following measure-and-prepare maps
\footnote{
A dynamical transformation \(\phi\) is said to be a measure-and-prepare map if it can be written as 
\begin{equation}
    \phi(p)=\sum_ie_i(p)p_i
\end{equation}
by using a family of states \(p_i\) and positive observables \(e_i\) satisfying \(\sum_ie_i=\mathbf{1}\).
}:
\begin{equation}
    \phi(p)=\frac{1+C_{\vec{Z}}(p)}{2}s_++\frac{1-C_{\vec{Z}}(p)}{2}s_-\label{eq:M-P}
\end{equation}
for a fixed $\vec{Z}\in\{X,Y\}^2$ and $s_\pm\in\ext(\NSstates)$.
These are trivial measure-and-prepare PSP maps, 
analogous to constant pure-state preparation channels in quantum theory.
Our first main result is the construction of
genuinely
nontrivial PSP maps.
\begin{thm}[PSP maps producing PR boxes]\label{thm:characterization}
Choose ordered pairs
\(\vec Z^{(+)}\neq\vec Z^{(-)}\in\{X,Y\}^2\),
\(\vec W^{(1)}\neq\vec W^{(2)}\in\{X,Y\}^2\), and a sign
\(\sigma\in\{\pm1\}\).  Define affine functions
\begin{equation}\label{eq:constructed-correlators}
r_{\vec Z}(p)\coloneqq
\begin{cases}
\sigma C_{\vec W^{(1)}}(p), & \vec Z=\vec Z^{(+)},\\
-\sigma C_{\vec W^{(1)}}(p), & \vec Z=\vec Z^{(-)},\\
\sigma C_{\vec W^{(2)}}(p), & \text{otherwise}.
\end{cases}
\end{equation}
Then the formula
\begin{equation}\label{eq:constructed-map}
\phi(p)(\vec m|\vec Z)\coloneqq\frac14\{1+m_Am_B r_{\vec Z}(p)\}
\end{equation}
defines an affine map \(\phi:\NSstates\to\NSstates\).  This map is PSP and, for
every pure state \(p\in\ext(\NSstates)\),
\begin{equation}
\phi(p)=p_{\vec r(p)}\in\Ent,
\end{equation}
where ${\vec r(p)}=(r_{\vec Z}(p))_{\vec{Z}\in\{X,Y\}^2}$.
In particular, \(\phi\) generates entanglement.  
\end{thm}

\begin{proof}

Since \(-1\le r_{\vec Z}(p)\le1\) for every \(p\in\NSstates\) and every \(\vec Z\), \(\phi(p)\) defined by Eq.~\eqref{eq:constructed-map} satisfies the condition of Eq.~\eqref{eq:positivity}.
Hence, \(\phi\) is a well-defined map on \(\NSstates\).
The affineness of \(\phi\) follows from the affineness of \(C_{\vec W^{(i)}}\) (\(i=1,2\)).

If \(p\) is pure, then \(C_{\vec W^{(1)}}(p),C_{\vec W^{(2)}}(p)\in\{\pm1\}\).
Therefore each \(r_{\vec Z}(p)\) is a sign and
\begin{equation}
\prod_{\vec Z}r_{\vec Z}(p)
=\bigl(\sigma C_{\vec W^{(1)}}(p)\bigr)
 \bigl(-\sigma C_{\vec W^{(1)}}(p)\bigr)
 \bigl(\sigma C_{\vec W^{(2)}}(p)\bigr)^2
=-1.
\end{equation}
Thus \(\phi(p)\) is the PR box \(p_{\vec r(p)}\).  This proves the PSP property
and entanglement generation. 
\end{proof}
This map is not measure-and-prepare; see 
the remark following the proof of Lemma~\ref{lem:heisenberg-correlator} in End Matter.

This theorem provides 
the first explicit example of nontrivial dynamics generating beyond-quantum entanglement.
This stands in sharp contrast to the no-go theorem of Gross
\emph{et al.}~\cite{Gross2010}, 
which shows that no reversible dynamics can generate entanglement in boxworld.
The physical significance of our result is 
that the PR boxes are not merely present in the state space: they are dynamically accessible from uncorrelated states through transformations that are irreversible yet pure-state-preserving.

Since different choices of $\vec{Z}^{(\pm)},\vec{W}^{(i)},$ and $\sigma$ define different maps, 
we have \(4\cdot3\cdot4\cdot3\cdot2=288\) maps.
Furthermore, we 
completely identify the PSP maps on \(\NSstates\)
and find 3032 PSP maps, including disentangling PSP maps that are similar to the above entangling ones; see End Matter for the complete list.

\paragraph{Completeness of the construction} 
We now show that the maps constructed in Theorem~\ref{thm:characterization} exhaust all nontrivial entangling PSP maps. 
\begin{thm}[Completeness of entangling PSP maps]\label{thm:completeness} Every entangling PSP map on \(\NSstates\) 
that is not measure-and-prepare is one of the \(288\) maps constructed in Theorem~\ref{thm:characterization}. 
\end{thm} 
The key to the proof is to work in the Heisenberg picture. 
Let \(\Observables\) denote the real vector space of affine functions \(\NSstates\to\mathbb R\) and call its elements observables. 
\(\Observables\) is generated by the expectation values $A_Z,B_Z,C_{\vec{Z}}$ and the constant observable \(\mathbf1:p\mapsto1\).
We define the Heisenberg picture \(\phi^\dagger\) of a dynamical transformation \(\phi\) as \(\phi^\dagger: O\mapsto O\circ\phi\) for \(O\in\Observables\).
Thus, instead of  focusing on states under \(\phi\), we focus on the output correlators as affine observables on the input state space.

Pure-state preservation imposes a strong restriction on these observables. 
\begin{lem}\label{lem:heisenberg-correlator} 
For every PSP map \(\phi\) and every output setting \(\vec Z\in\{X,Y\}^2\), 
\begin{equation}\label{eq:heisenberg-correlator-classification} 
\phi^\dagger C_{\vec Z} \in \{\pm\mathbf1\} \cup \{\pm C_{\vec W}:\vec W\in\{X,Y\}^2\}.
\end{equation} 
\end{lem} 
In other words, each output correlation is either fixed to a constant sign or is equal, 
up to a sign, to one of the four input correlators. 

If at most one correlator \(C_{\vec W}\) appears among the observables \(\phi^\dagger C_{\vec Z}\), 
the output is determined by the single observable \(C_{\vec W}\). 
In this case, the map is measure-and-prepare.
Otherwise, at least two distinct correlators appear.
In this case, 
the entangling PSP condition forces every pure input state to be mapped to a PR box. 
Therefore 
\begin{equation}\label{eq:product-minus-one} 
\prod_{\vec Z\in\{X,Y\}^2} \phi^\dagger C_{\vec Z}(p) = -1 \qquad \text{for every }p\in\ext(\NSstates). 
\end{equation}
Together with Lemma~\ref{lem:heisenberg-correlator}, 
this implies that each input correlator must occur with even multiplicity in the product in Eq.~\eqref{eq:product-minus-one}, which 
is exactly
the correlator pattern of Theorem~\ref{thm:characterization}.
A complete proof is given in End Matter;
see also Supplemental Material~\cite{SM} for technical details,
where we further provide a complete classification of all PSP maps, including the nonentangling ones.

\paragraph{Postselected circuit realization}
Postselection has previously been used to reproduce exact PR-box statistics using quantum or even classical resources~\cite{Marcovitch2007,Plavala2020,Bacciagaluppi2021}.
Here, we show that it can implement not merely a fixed PR-box distribution, 
but the entire entangling map \(\phi\): 
for every input state \(p\), the same circuit
realizes \(\phi(p)\) upon success, with probability \(1/2\) independently of \(p\)~\footnote{
Here we distinguish the application of the map from the subsequent measurement of the output state. Once the map has been applied, the output state can be measured locally, without either party revealing its measurement setting to the other. This is distinct from measurement-dependent wirings~\cite{Barrett2007}, in which the two local settings must be jointly available during the transformation.
}.

First, we encode each no-signaling state \(p\in\NSstates\) as a joint distribution of classical random bits as follows. 
Let \(\hat Z_A\) and \(\hat Z_B\) be independent uniformly random bits, where \(0\) and \(1\) are identified with the measurement settings \(X\) and \(Y\), respectively.
Conditioned on \(\hat Z_A=Z_A\) and \(\hat Z_B=Z_B\), let \(\hat m_A\) and \(\hat m_B\) be sampled according to \(p(m_A,m_B|Z_A,Z_B)\). 
The resulting joint distribution is
\begin{align}
&\widehat p(\hat m_A=m_A,\hat m_B=m_B,\hat Z_A=Z_A,\hat Z_B=Z_B)\nonumber\\
&\coloneqq\frac14 p(m_A,m_B|Z_A,Z_B).  \label{eq:encoding}
\end{align}
These random bits simulate the original state \(p\) as follows. 
For chosen settings \(\vec Z=(Z_A,Z_B)\), each party \(C\in\{A,B\}\) reads out \((\hat m_C,\hat Z_C)\) and retains \(\hat m_C\) only if \(\hat Z_C=Z_C\). The run is thus retained with probability \(1/4\), and, conditioned on retention, \((\hat m_A,\hat m_B)\) is distributed according to the original \(p\).

Next, we construct a classical logical circuit that simulates 
the action of the entangling map \(\phi\) at the level of encoded random bits as
\begin{equation}
\raisebox{-4.5em}
{\includegraphics[width=0.8\linewidth]{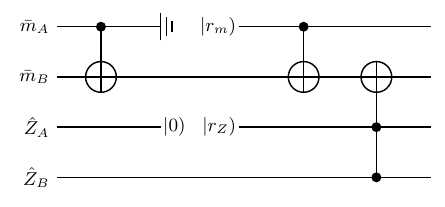}},
\end{equation}
where time flows from left to right.
This circuit acts on samples from \(\widehat p\) with the convention $\hat{m}_A=(-1)^{\bar{m}_A}$ and $\hat{m}_B=(-1)^{\bar{m}_B}$.
Here, 
two- and three-bit gates are the controlled-not and Toffoli gates, respectively;
the three-line termination symbol denotes the discarding of a bit;
$\ket{r_m}-$ and $\ket{r_Z}-$ denote the preparation of independent uniformly random bits;
and $-\ket{0}$ denotes postselection on $0$, meaning that an output of $0$ is retained, whereas an output of $1$ is discarded.

The resulting unnormalized output distribution \(\hat{p}'\) is
\begin{align}
&\hat{p}'(m_A,m_B,Z_A,Z_B)\nonumber\\
&=
\frac{1}{4}\sum_{\tau\in\{\pm1\}}
\hat{p}\bigl(
\tau,
(-1)^{Z_AZ_B}\tau m_Am_B,
X,
Z_B
\bigr)
\nonumber\\
&=
\frac{1}{2}
\widehat{\phi(p)}(m_A,m_B,Z_A,Z_B),
\label{eq:implementation_PSP}
\end{align}
where $\phi$ is the entangling PSP map defined in Theorem~\ref{thm:characterization}, with
$\vec{Z}^{(+)}=\vec{W}^{(1)}=(X,Y)$,
$\vec{Z}^{(-)}=(Y,Y)$,
$\vec{W}^{(2)}=(X,X)$,
and $\sigma=+1$.
Therefore, Eq.~\eqref{eq:implementation_PSP} shows that, conditioned on successful postselection, the normalized output distribution is $\widehat{\phi(p)}$; hence, the circuit implements the entangling PSP map $\phi$ at the level of the encoding in Eq.~\eqref{eq:encoding}, 
with a success probability of $1/2$.
Similar constructions apply to general choices of
$\vec{Z}^{(\pm)}$, $\vec{W}^{(1,2)}$, and $\sigma\in\{\pm1\}$.

\paragraph{Conclusions}
We have provided, to our knowledge, the first explicit dynamics that generates beyond-quantum entanglement. The constructed pure-state-preserving transformations on 
\(\NSstates\) map deterministic product states to PR boxes and, more strongly, send every pure input to a PR box. We identified 288 non-measure-and-prepare entangling transformations and proved that they exhaust this class. These results show that 
\(\NSstates\) is not dynamically sterile: its maximally nonlocal states can be reached from uncorrelated states without introducing mixing. The transformations are necessarily irreversible, so our construction sharpens, rather than contradicts, the known no-go theorem for reversible dynamics. Determining which physical requirements exclude such entangling transformations may help clarify why nature realizes quantum entanglement rather than stronger entanglement in no-signaling theories.

Several directions remain open.
First, it would be important to clarify the information-processing power enabled by the entangling maps constructed here.
At the level of static resources, freely available PR boxes are known to trivialize two-party communication complexity~\cite{vanDam1999,Brassard2006}.
Our result raises the corresponding dynamical question: how does the situation change when PR boxes must be generated from uncorrelated states rather than supplied as static resources, and what resource cost should be assigned to the transformations that generate them?
Closely related is the operational implementation of these transformations. Although we have presented a postselection protocol, the optimal implementation remains unknown.

Finally, although the triviality of reversible dynamics extends to multipartite boxworld~\cite{Gross2010,Al-Safi2014}
and other systems~\cite{Masanes2014,Al-Safi2015}, 
PSP dynamics may exhibit qualitatively new behavior in larger systems.
Already in tripartite boxworld, the state space contains genuinely multipartite extremal boxes with no bipartite counterpart~\cite{Pironio2011}.
It remains to determine which such boxes are reachable from product states by PSP transformations and whether the rigidity found here in the bipartite setting persists in general multipartite systems.

\paragraph{Acknowledgments}
S.U.~and A.H.~were supported by the FoPM WINGS Program at the University of Tokyo.
A.H.~was also supported by KAKENHI
Grant No.~JP25KJ0833 from the Japan Society for the
Promotion of Science (JSPS) and by the JSR Fellowship at the University of Tokyo.
H.A.~was supported by JSPS KAKENHI
Grant Nos.~25KJ0043 and 26K17031.
K.T. was supported by JST PRESTO (Grants Nos. JPMJPR2256 and JPMJPR2596) and JSPS KAKENHI (Grants Nos. JP23K17664, JP25K17312,  JP26H00385, and JP26H00382).

\bibliography{applied-GPT_ref,suppl}

\clearpage
\section*{End Matter}
For simplicity, we denote \(\overline A_Z\coloneqq\phi^\dagger A_Z\), \(\overline B_Z\coloneqq\phi^\dagger B_Z\), and \(\overline C_{\vec Z}\coloneqq\phi^\dagger C_{\vec Z}\) for \(Z\in\{X,Y\}\) and \(\vec Z\in\{X,Y\}^2\).
\begin{proof}[Proof of Lemma~\ref{lem:heisenberg-correlator}]
It suffices to show that the only observables \(O\in \Observables\) taking values \(\pm1\) on all pure states are either \(\pm\mathbf1\) or \(\pm C_{\vec W}\) (\(\vec W\in\{X,Y\}^2\)).
Indeed, since \(\phi\) is PSP, every observable of the form \(\phi^\dagger C_{\vec Z}\) must take values \(\pm1\) on all pure states, and the claim follows.
We write \(O\) in expectation coordinates:
\begin{equation}\label{eq:EM-correlator-expansion}
O
=
c\mathbf1
+
\sum_{U\in\{X,Y\}}\alpha_U A_U
+
\sum_{V\in\{X,Y\}}\beta_V B_V
+
\sum_{\vec W\in\{X,Y\}^2}
\gamma_{\vec W}C_{\vec W}
\end{equation}
with real $c,\alpha_U,\beta_V,$ and $\gamma_{\vec W}$. 

By assumption, \(O(p)\in\{\pm 1\}\) for every pure state \(p\), 
and hence
\begin{equation}\label{eq:EM-square-one}
\bigl(
O(p)-\mathbf1
\bigr)
\bigl(
O(p)+\mathbf1
\bigr)
=
0
\end{equation}
on every pure input state \(p\).

Let
\(\langle\cdot\rangle_{\Sep}\)
denote the uniform average over the sixteen deterministic vertices.
The functions
\(
\mathbf1,\ 
A_U,\ 
B_V,\ 
C_{\vec W}
\)
are mutually orthogonal under this average,
because $A_U$ and $B_V$ are independent random signs and 
$C_{\vec W}=A_{W_A}B_{W_B}$ on $\Sep$.
Averaging Eq.~\eqref{eq:EM-square-one} over
\(\Sep\) gives
\begin{equation}\label{eq:EM-deterministic-average}
0
=
c^2
+
\sum_U\alpha_U^2
+
\sum_V\beta_V^2
+
\sum_{\vec W}\gamma_{\vec W}^2
-
1.
\end{equation}

Let
\(\langle\cdot\rangle_{\Ent}\)
denote the uniform average over the eight PR vertices.
All local marginals vanish on these vertices, while the four correlators are
mutually orthogonal, 
because any two correlators are independent random signs.
Averaging Eq.~\eqref{eq:EM-square-one} over
\(\Ent\) gives
\begin{equation}\label{eq:EM-PR-average}
0
=
c^2
+
\sum_{\vec W}\gamma_{\vec W}^2
-
1.
\end{equation}
Subtracting
Eq.~\eqref{eq:EM-PR-average}
from
Eq.~\eqref{eq:EM-deterministic-average},
we obtain
\begin{equation}
\sum_U\alpha_U^2+\sum_V\beta_V^2=0.
\end{equation}
Hence all local-marginal coefficients vanish:
\begin{equation}
O
=
c\mathbf1
+
\sum_{\vec W}\gamma_{\vec W}C_{\vec W}.
\end{equation}

For any $\vec W$ and any pure state \(p\), there exists a pure state \(p'\) such that \((C_{\vec Z}(p'))_{\vec Z}\) differs from \((C_{\vec Z}(p))_{\vec Z}\) only in the sign of \(C_{\vec W}\).
For such a pair of states,
\[
\abs{O(p)-O(p')}
=
2\abs{\gamma_{\vec W}}.
\]
Since both \(O(p)\) and \(O(p')\) take values in \(\{\pm1\}\), it follows that \(\gamma_{\vec W}\in\{0,\pm1\}\).

Moreover, if there were two or more \(\vec W\) satisfying \(\gamma_{\vec W}=\pm1\), then one could choose pure states \(p\) and \(p'\) for which the signs of 
two corresponding correlators are flipped, 
yielding \(\abs{O(p)-O(p')}=4\).
This is impossible because \(O(p),O(p')\in\{\pm1\}\).
Therefore, there exists at most one \(\vec W\) such that \(\gamma_{\vec W}=\pm1\).

If all of \(\gamma_{\vec W}\) vanish, Eq.~\eqref{eq:EM-PR-average} gives
\(c=\pm1\).
If \(\gamma_{\vec W}=\pm1\) for exactly one \(\vec W\), then Eq.~\eqref{eq:EM-PR-average} implies \(c=0\).
Therefore
\begin{equation}
O
\in
\{\pm\mathbf1\}
\cup
\{\pm C_{\vec W}:\vec W\in\{X,Y\}^2\},
\end{equation}
as claimed.
\end{proof}

The above proof implies that the only effects taking values in \(\{0,1\}\) on all pure states are
\(\mathbf0\), \(\mathbf1\), and \((\mathbf1\pm C_{\vec Z})/2\).
It follows that the only PSP measure-and-prepare maps are those of Eq.~\eqref{eq:M-P}.
Indeed, without loss of generality, we may assume that distinct measurement outcomes correspond to distinct prepared pure states.
Then, if a measurement effect took a value other than \(0\) or \(1\) on some pure state, the corresponding output would be a nontrivial convex combination of distinct prepared pure states, contradicting the PSP property.
Therefore, the maps constructed in Theorem~\ref{thm:characterization} are not measure-and-prepare.

\begin{proof}[Proof of Theorem~\ref{thm:completeness}]
For an input setting pair \(\vec W\), define
\begin{equation}\label{eq:EM-local-subspace}
\mathcal V_{\vec W}
\coloneqq
\operatorname{span}
\{\mathbf1,A_{W_A},B_{W_B},C_{\vec W}\}.
\end{equation}

We use the following  three facts:
First, if at most one distinct input correlator occurs among the nonconstant
observables
\(
\overline C_{\vec Z}
\),
then the 
PSP map $\phi$
is measure-and-prepare.
Second, if an output correlator is represented by a signed input correlator,
\begin{equation}\label{eq:EM-marginal-support-assumption}
\overline C_{\vec Z}
=
\pm C_{\vec W},
\end{equation}
then the corresponding local marginals satisfy
\begin{equation}\label{eq:EM-marginal-support}
\overline A_{Z_A},\overline B_{Z_B}
\in
\mathcal V_{\vec W}.
\end{equation}
Third, if an output local marginal \(\overline M\) with \(M\in\{A_X,A_Y,B_X,B_Y\}\) satisfies
\begin{equation}
\overline M\in\mathcal{V}_{\vec W}\cap \mathcal{V}_{\vec W'}\quad (\vec W\neq \vec W'),
\end{equation}
then 
\begin{equation}
\overline M\in\{\mathbf0,\pm\mathbf1 ,\pm N\},
\end{equation}
where \(N=A_{W_A}\) if \(W_A=W_A'\),  \(N=B_{W_B}\) if \(W_B=W_B'\) 
and \(N=\bm0\) otherwise.
Here, \(\bm0=0\times\bm 1\) is the constant observable \(p\mapsto 0\).
See Supplemental Material~\cite{SM} for the proofs of these facts.

Let \(\phi\) be an entangling PSP map that is not
measure-and-prepare.
By the first fact, 
more than one distinct input correlator
occurs among the four observables
\(
\overline C_{\vec Z}
\).

Arrange the four observables according to their output settings:
\begin{equation}\label{eq:EM-output-square}
\begin{pmatrix}
\overline C_{X,X} & \overline C_{X,Y}\\
\overline C_{Y,X} & \overline C_{Y,Y}
\end{pmatrix}.
\end{equation}

Suppose first that two neighboring entries involve distinct input
correlators.
Without loss of generality, write
\begin{equation}\label{eq:EM-neighboring-correlators}
\overline C_{X,X}
=
\pm C_{\vec W^{(1)}},
\quad
\overline C_{X,Y}
=
\pm C_{\vec W^{(2)}},
\quad
\vec W^{(1)}\neq\vec W^{(2)}.
\end{equation}
The two output correlators share Alice's setting 
\(X\).
From the second fact,
\begin{equation}\label{eq:EM-neighboring-intersection}
\overline A_X
\in
\mathcal V_{\vec W^{(1)}}
\cap
\mathcal V_{\vec W^{(2)}}.
\end{equation}

Since \(\phi\) is entangling, 
some \(p_0\in\Sep\) is mapped to \(\phi(p_0)\in\Ent\).
At this input,
\begin{equation}
\overline A_{X}(p_0)=0,
\qquad
N(p_0)\in\{\pm1\},
\end{equation}
for every \(N\in\{A_X,A_Y,B_X,B_Y\}\).
Among the choices allowed by the third fact, this is only possible for
\begin{equation}\label{eq:EM-vanishing-local-marginal}
\overline A_{X}=\textbf0.
\end{equation}
Every pure output state therefore has a vanishing local marginal.
A deterministic pure state cannot have such a marginal.
Hence every pure input is mapped to a PR box.

It remains to exclude the case in which no neighboring entries of
Eq.~\eqref{eq:EM-output-square}
involve distinct input correlators.
Since more than one distinct input correlator occurs, two distinct
correlators must then appear at opposite corners.
After relabelling the output settings, write
\begin{equation}\label{eq:EM-diagonal-correlators}
\overline C_{X,X}
=
\pm C_{\vec W^{(1)}},
\quad
\overline C_{Y,Y}
=
\pm C_{\vec W^{(2)}},
\quad
\vec W^{(1)}\neq\vec W^{(2)}.
\end{equation}
The two remaining entries must be constant signs:
\begin{equation}\label{eq:EM-diagonal-constants}
\overline C_{X,Y}=\delta\mathbf1,
\quad
\overline C_{Y,X}=\delta'\mathbf1
\quad
(\delta,\delta'\in\{\pm1\}).
\end{equation}

The first identity in Eq.~\eqref{eq:EM-diagonal-constants} means that every output state $q=\phi(p)$ satisfies
\(
C_{X,Y}(q)=\delta
\).
Since 
Eq.~\eqref{eq:positivity} implies 
the relation \(A_X(q)=\delta B_Y(q)\),
we have
\begin{equation}\label{eq:EM-bridge-marginal}
\overline A_X
=
\delta\,\overline B_Y.
\end{equation}
Using
Eq.~\eqref{eq:EM-marginal-support}
for the two nonconstant observables in
Eq.~\eqref{eq:EM-diagonal-correlators},
we obtain \(\overline A_X\in\mathcal V_{\vec W^{(1)}}\), \(\overline B_Y\in\mathcal V_{\vec W^{(2)}}\).
Together with
Eq.~\eqref{eq:EM-bridge-marginal},
this gives
\begin{equation}
\overline A_X
\in
\mathcal V_{\vec W^{(1)}}
\cap
\mathcal V_{\vec W^{(2)}}.
\end{equation}
The argument leading to
Eq.~\eqref{eq:EM-vanishing-local-marginal}
applies again.
Hence every pure input is mapped to a PR box.

We have proved that every pure input state is mapped to a PR box.
Consequently, we have $\overline A_U=\overline B_V=\bm0$ and
\begin{equation}\label{eq:EM-product-minus-one}
\prod_{\vec Z\in\{X,Y\}^2}
\overline C_{\vec Z}(p)
=
-1
\qquad
\text{for every }p\in\ext(\NSstates).
\end{equation}

By Lemma~\ref{lem:heisenberg-correlator}, the left-hand side of
Eq.~\eqref{eq:EM-product-minus-one}
is a signed monomial in the input correlators.
Since the pure correlation vectors cover the full sign cube \(\{\pm1\}^4\), 
this monomial
can be constant only if every input correlator occurs with even
multiplicity.
 
Since more
than one distinct input correlator occurs, 
exactly two distinct input correlators occur, each exactly twice, and
no constant factor occurs.
We note that the second case considered above is therefore ruled out.
To realize the  overall sign  \(-1\) in
Eq.~\eqref{eq:EM-product-minus-one},
one of the two correlators  
must occur with opposite signs, 
whereas the other 
must occur with the same sign twice.

Let
\(C_{\vec W^{(1)}}\)
denote the correlator occurring with opposite signs, and let
\(C_{\vec W^{(2)}}\)
denote the correlator occurring with equal signs.
Let \(\sigma\) be the repeated sign of
\(C_{\vec W^{(2)}}\).
Among the two occurrences of
\(C_{\vec W^{(1)}}\),
denote by
\(\vec Z^{(+)}\)
the one whose sign is \(\sigma\), and by
\(\vec Z^{(-)}\)
the one whose sign is \(-\sigma\).
Then
\begin{equation}\label{eq:EM-final-correlator-pattern}
\overline C_{\vec Z}
=
\begin{cases}
 \sigma C_{\vec W^{(1)}},
 & \vec Z=\vec Z^{(+)},\\
-\sigma C_{\vec W^{(1)}},
 & \vec Z=\vec Z^{(-)},\\
 \sigma C_{\vec W^{(2)}},
 & \text{otherwise}.
\end{cases}
\end{equation}

This is precisely the affine map constructed in Theorem~\ref{thm:characterization}.
\end{proof}

\paragraph{Comprehensive list of PSP maps}
Here, we provide a complete list of PSP maps, not necessarily entangling.
First, there are \(8\cdot8\cdot2=128\) reversible maps consisting of relabelings of inputs, outputs, and subsystems \cite{Gross2010}.
We also have \(4\cdot 24\cdot 23=2208\) two-outcome measure-and-prepare maps defined by Eq.~\eqref{eq:M-P} for \(s_+\neq s_-\) and \(24\) constant maps defined by
\begin{equation}
    \phi(p)=s
\end{equation}
for \(s\in \ext(\NSstates)\).

In addition to the entangling maps defined by Theorem~\ref{thm:characterization}, there are the following \(4\cdot3\cdot2^4\cdot2=384\) nontrivial PSP maps:
Choose an ordered pair \(\vec W^{(X)}\neq\vec W^{(Y)}\in\{X,Y\}^2\) and sign vectors \(\vec \alpha=(\alpha_X,\alpha_Y),\; \vec \beta=(\beta_X,\beta_Y)\in\{\pm1\}^2\).
Then, the maps defined by
\begin{equation}\label{eq:SEP-PSP-A}
    \begin{aligned}
        \phi_A(p)(\vec m \mid \vec Z)
        &=\frac{1}{4}\{1+m_A\alpha_{Z_A}+m_B\beta_{Z_B}C_{\vec W^{(Z_B)}}(p)\\
        &\quad\quad\quad +m_Am_B\alpha_{Z_A}\beta_{Z_B}C_{\vec W^{(Z_B)}}(p)\},
    \end{aligned}
\end{equation}
\begin{equation}\label{eq:SEP-PSP-B}
    \begin{aligned}
        \phi_B(p)(\vec m \mid \vec Z)
        &=\frac{1}{4}\{1+m_A\alpha_{Z_A}C_{\vec W^{(Z_A)}}(p)+m_B\beta_{Z_B}\\
        &\quad\quad\quad +m_Am_B\alpha_{Z_A}\beta_{Z_B}C_{\vec W^{(Z_A)}}(p)\}
    \end{aligned}
\end{equation}
are PSP maps.
These $\phi_{A,B}$ map every pure state into a separable pure state (cf.~Eq.~\eqref{eq:characterization_sep_pure}).
Together with the reversible, measure-and-prepare, constant, and entangling maps listed above, they exhaust all PSP maps on \(\NSstates\); see Supplemental Material~\cite{SM} for the proof.

\clearpage
\clearpage\clearpage
\makeatletter
   	\c@secnumdepth=4
\makeatother

\setcounter{equation}{0}
\setcounter{figure}{0}
\setcounter{section}{0}
\setcounter{table}{0}
\setcounter{thm}{0} 
\renewcommand{\theequation}{S\arabic{equation}}
\renewcommand{\thefigure}{S\arabic{figure}}
\renewcommand{\theHequation}{\theequation}
\renewcommand{\theHfigure}{\thefigure}
\renewcommand{\thethm}{S\arabic{thm}}
\renewcommand{\thesection}{S\arabic{section}}

\renewcommand{\bibnumfmt}[1]{[S#1]}
\renewcommand{\citenumfont}[1]{S#1}

\renewcommand{\thepage}{S\arabic{page}}
\setcounter{page}{1}

\makeatletter
\let\orig@footnote\footnote
\long\def\SM@pagefootnote#1{%
  \refstepcounter{footnote}%
  \begingroup
    \protected@xdef\@thefnmark{\thefootnote}%
  \endgroup
  \@footnotemark
  \@footnotetext{#1}%
}
\let\footnote\SM@pagefootnote
\makeatother

\title{
  Supplemental Material for \protect\\
  ``\nameoftitle''
}

\maketitle
\onecolumngrid
This Supplemental Material provides proofs of the technical facts used in the End Matter to establish that the identified family exhausts all entangling PSP maps, and presents the complete classification of all PSP maps.

For a given dynamical transformation $\phi$,
we denote \(\overline A_Z\coloneqq\phi^\dagger A_Z\), \(\overline B_Z\coloneqq\phi^\dagger B_Z\), and \(\overline C_{\vec Z}\coloneqq\phi^\dagger C_{\vec Z}\) for \(Z\in\{X,Y\}\) and \(\vec Z\in\{X,Y\}^2\).
For \(Z\in\{X,Y\}\), denote the other setting by 
\(\check{Z}\), i.e., 
\(\check{X}=Y\) and \(\check{Y}=X\).

\section{Technical details for the proof of the completeness of entangling PSP maps}
We prove three facts that are used in the proof of Theorem~\ref{thm:completeness} in End Matter. 
Note that the following statements hold without assuming that the map \(\phi\) is entangling.

Let us recall the positivity conditions (Eq.~\eqref{eq:positivity} in the main text):
\begin{equation}
    \abs{A_{Z_A}-\delta B_{Z_B}}\le 1- \delta C_{\vec{Z}},\ \delta=\pm1,
    \label{eq:SM-positivity}
\end{equation}
which completely determine $\NSstates$ with $A_U,B_V,C_{\vec Z}\in[-1,1]$.
For $\vec W\in\{X,Y\}^2$ and $\delta\in\{\pm1\}$, we define 
\begin{equation}
F_{\vec W}^{\delta}
\coloneqq
\{
p\in\NSstates:
C_{\vec W}(p)=\delta
\}.
\end{equation}
An affine observable \(O\in\Observables\) vanishing on
\(F_{\vec W}^{\delta}\) satisfies
\begin{equation}\label{eq:SM-face-vanishing}
O
\in
\operatorname{span}
\{
\mathbf1-\delta C_{\vec W},
\;
A_{W_A}-\delta B_{W_B}
\}.
\end{equation}
To show this, 
we take the following coordinates:
\begin{align}
    A_{W_A}=\delta B_{W_B}&=\epsilon;\\
    A_{\check{W}_A}&=\epsilon_{\check{W}_A};\\
    B_{\check{W}_B}&=\epsilon_{\check{W}_B};\\
    C_{\vec W}&=\delta;\\
    C_{\vec W'}&=\epsilon_{\vec W'}\quad(\vec W'\neq \vec W),
\end{align}
where $\epsilon,\epsilon_{\check{W}_A},\epsilon_{\check{W}_B}$ and $\epsilon_{\vec W'}$ are
arbitrary real numbers in $[-1/3,1/3]$.
These coordinates satisfy Eq.~\eqref{eq:SM-positivity}, 
and therefore define some $p\in F_{\vec W}^{\delta}$.
By expanding \(O\) in the eight coordinates, we have
\begin{equation}
O
=
c\mathbf1
+
\sum_{U\in\{X,Y\}}\alpha_U A_U
+
\sum_{V\in\{X,Y\}}\beta_V B_V
+
\sum_{\vec W'\in\{X,Y\}^2}
\gamma_{\vec W'}C_{\vec W'}.
\end{equation}
By the vanishing assumption, we have
\begin{equation}
    0=c+(\alpha_{W_A}+\delta\beta_{W_B})\epsilon+
\alpha_{\check{W}_A}\epsilon_{\check{W}_A}+\beta_{\check{W}_B}\epsilon_{\check{W}_B}+\gamma_{\vec W}\delta+\sum_{\vec W'\in\{X,Y\}^2\setminus\{\vec W\}}\gamma_{\vec W'}\epsilon_{\vec W'}.
\end{equation}
Since $\epsilon$'s can be taken arbitrarily, we have
\begin{align}
    c+\gamma_{\vec W}\delta=0,\quad \alpha_{W_A}+\delta\beta_{W_B}=0,\quad 
    \alpha_{\check{W}_A}=\beta_{\check{W}_B}=\gamma_{\vec W'}=0,
\end{align}
which yields Eq.~\eqref{eq:SM-face-vanishing}.

We are now in a position to show the second fact used in End Matter.
\begin{lem}\label{lem:SM-marginal-support}
Suppose that
\begin{equation}\label{eq:SM-marginal-support-assumption}
\overline{C}_{U,V}
=
\delta C_{\vec W},
\qquad
\delta\in\{\pm1\}.
\end{equation}
Then
\begin{equation}\label{eq:SM-marginal-support-conclusion}
\overline{A}_U,
\overline{B}_V
\in
\mathcal V_{\vec W}
\coloneqq
\operatorname{span}
\{\mathbf1,A_{W_A},B_{W_B},C_{\vec W}\}.
\end{equation}
\end{lem}

\begin{proof}
Applying Eq.~\eqref{eq:SM-positivity} to the output state gives
\begin{equation}
|
\overline{A}_U-\overline{B}_V
|
\le
1-\delta C_{\vec W},
\qquad
|
\overline{A}_U+\overline{B}_V
|
\le
1+\delta C_{\vec W}.
\end{equation}
Hence
\(
\overline{A}_U-\overline{B}_V
\)
vanishes on
\(
F_{\vec W}^{\delta}
\),
whereas
\(
\overline{A}_U+\overline{B}_V
\)
vanishes on
\(
F_{\vec W}^{-\delta}
\).
By Eq.~\eqref{eq:SM-face-vanishing},
\begin{align}
\overline{A}_U-\overline{B}_V
&\in
\operatorname{span}
\{
\mathbf1-\delta C_{\vec W},
\;
A_{W_A}-\delta B_{W_B}
\},
\\
\overline{A}_U+\overline{B}_V
&\in
\operatorname{span}
\{
\mathbf1+\delta C_{\vec W},
\;
A_{W_A}+ \delta B_{W_B}
\}.
\end{align}
Taking the sum and difference proves the claim.
\end{proof}

Using this lemma, we can show the first fact used in End Matter.
\begin{lem}\label{lem:SM-one-correlator}
Let
\(
\phi:\NSstates\to\NSstates
\)
be a PSP map.
If at most one distinct input correlator occurs among the nonconstant
observables
\(
\overline{C}_{U,V}
\),
then \(\phi\) is measure-and-prepare.
\end{lem}

\begin{proof}
By Lemma~\ref{lem:heisenberg-correlator} in the main text, 
each output correlator is either a
constant sign or a signed input correlator.

If no input correlator occurs, the output correlation vector is constant.
If it is that of a PR box, there is a unique compatible pure state, so \(\phi\) is a constant preparation map.
If it is deterministic, there are exactly two compatible pure states, 
which are 
uniquely distinguished by the value of any one of the local observables, say \(A_X\).
Hence the output is completely determined by the binary value of \(\overline A_X\) at the input state, and 
\(\phi\) can be implemented by measuring \(\overline A_X\), 
followed by an appropriate state preparation.

Suppose that exactly one input correlator \(C_{S,T}\) occurs.
For a pure input \(p\), set
\[
s(p)\coloneqq C_{S,T}(p)\in\{\pm1\}.
\]
The output correlation vector depends only on \(s(p)\).

Choose an output setting \((U,V)\) such that
\(
\overline{C}_{U,V}=\pm C_{S,T}
\).
By Lemma~\ref{lem:SM-marginal-support},
\begin{equation}\label{eq:SM-single-correlator-marginal}
\overline{A}_U=a\mathbf1+bA_S+cB_T+dC_{S,T}.
\end{equation}

If the output correlation vector corresponding to \(s(p)=+1\) is that of a PR box,
evaluating Eq.~\eqref{eq:SM-single-correlator-marginal} on PR and deterministic inputs with \(s(p)=+1\) yields
\begin{equation}
a+d=a\pm(b+c)+d=0.
\end{equation}
Hence,
\begin{equation}
b+c=0.
\end{equation}

On the other hand, suppose that the output correlation vector corresponding to \(s(p)=+1\) is deterministic.
Evaluating Eq.~\eqref{eq:SM-single-correlator-marginal} on PR and deterministic inputs with \(s(p)=+1\) gives
\begin{equation}
a+d,\;a+b+c+d,\;a-b-c+d\in\{\pm1\}.
\end{equation}
If \(a+d=+1\), the second condition implies \(b+c=0\) or \(-2\), while the third condition further forces \(b+c=0\).
Similarly, if \(a+d=-1\), one again obtains \(b+c=0\).

By an analogous argument, whether the output correlation vector corresponding to \(s(p)=-1\) is that of a PR box or deterministic, one finds
\begin{equation}
b-c=0.
\end{equation}

Therefore, in all cases,
\begin{equation}
\overline{A}_U=a\mathbf1+dC_{S,T}.
\end{equation}
Similarly, we have \(\overline B_V\in\operatorname{span}\{\bm1,C_{S,T}\}\).
For the setting \(\check U(\neq U)\), 
\(\overline C_{\check U,V}\) is either \(\pm C_{S,T}\) or \(\pm\bm1\).
In the former case, the same argument yields
\(
\overline A_{\check U}
\in\operatorname{span}\{\bm1,C_{S,T}\}.
\)
In the latter case, Eq.~\eqref{eq:SM-positivity} implies
\(
\overline A_{\check U}=\pm\overline B_V.
\)
Since
\(\overline B_V\in\operatorname{span}\{\bm1,C_{S,T}\}\),
we again obtain
\(
\overline A_{\check U}
\in\operatorname{span}\{\bm1,C_{S,T}\}.
\)
The same argument applies to \(\overline B_{\check{V}}\) and we have
\begin{equation}
    \overline{A}_X,\overline{A}_Y,\overline{B}_X,\overline{B}_Y\in\operatorname{span}\{\bm1,C_{S,T}\}.
\end{equation}
Therefore, every coordinate of the output state \(\phi(p)\) is completely determined by \(C_{S,T}(p)\). 
Hence, \(\phi\) can be implemented by measuring \(C_{S,T}\), 
followed by an appropriate state preparation.
\end{proof}

We conclude this section with the proof of the third fact used in End Matter.
\begin{lem}\label{lem:possible_local_marginals}
Let
\(
\phi:\NSstates\to\NSstates
\)
be a PSP map.
Suppose that \(\overline M\) with \(M\in \{A_X,A_Y,B_X,B_Y\}\) satisfies
\begin{equation}
\overline M\in\mathcal{V}_{\vec W}\cap \mathcal{V}_{\vec W'}\quad (\vec W\neq \vec W').
\end{equation}
Then
\begin{equation}
\overline M\in\{\mathbf0,\pm\mathbf1 ,\pm N\},
\end{equation}
where \(N=A_{W_A}\) if \(W_A=W_A'\),  \(N=B_{W_B}\) if \(W_B=W_B'\) 
and \(N=\bm0\) otherwise.
\end{lem}
\begin{proof}
We first consider the case in which \(W_A=W_A'\).
From the assumption, we can write
\begin{equation}\label{eq:EM-local-affine-form}
\overline M=c\mathbf1+dA_{W_A}
\end{equation}
for some real numbers \(c,d\).

On pure input states, a local marginal takes values in
\(\{-1,0,+1\}\):
it is \(\pm1\) on deterministic boxes and \(0\) on PR boxes.
Since \(\phi\) is PSP, the same is true of \(\overline M\).
The only affine functions
\begin{equation}
t\longmapsto c+dt
\end{equation}
mapping
\(\{-1,0,+1\}\)
into itself are limited to $0,\pm1$ and $\pm t$.  
Therefore, we have
\begin{equation}
\overline M\in\{\mathbf0,\pm\mathbf1 ,\pm  A_{W_A}\}.
\end{equation}
The other cases can be shown in the same way.
\end{proof}

\section{Proof of the completeness of all PSP maps}

We say that a PSP map \(\phi\) destroys entanglement or is disentangling if \(\phi(\Ent)\cap\Sep\neq\varnothing\).
As in the entanglement generation case, the following holds.
\begin{prop}[Completeness of disentangling PSP maps]\label{prop:EDcompleteness} 
Every disentangling PSP map on \(\NSstates\) that is not measure-and-prepare is one of the \(384\) maps of Eq.~\eqref{eq:SEP-PSP-A} and Eq.~\eqref{eq:SEP-PSP-B} in End Matter.
\end{prop} 
\begin{proof}
Let \(\phi\) be a disentangling PSP map that is not measure-and-prepare.
By Lemma~\ref{lem:SM-one-correlator}, more than one distinct input correlator
occurs among the four observables \(\overline C_{\vec Z}\).

Arrange the four observables according to their output settings:
\begin{equation}\label{eq:SM-output-square}
\begin{pmatrix}
\overline C_{X,X} & \overline C_{X,Y}\\
\overline C_{Y,X} & \overline C_{Y,Y}
\end{pmatrix}.
\end{equation}

We can exclude the case in which no neighboring entries of Eq.~\eqref{eq:SM-output-square} involve distinct input correlators by exactly the same argument as in the entangling case.
Without loss of generality, write
\begin{equation}\label{eq:SM-output-correlator}
\overline C_{X,X}=\pm C_{\vec W^{(X)}},
\qquad
\overline C_{X,Y}=\pm C_{\vec W^{(Y)}},
\qquad
\vec W^{(X)}\neq\vec W^{(Y)}.
\end{equation}
From Lemma~\ref{lem:SM-marginal-support},
\begin{equation}
\overline A_X
\in
\mathcal V_{\vec W^{(X)}}\cap\mathcal V_{\vec W^{(Y)}}.
\end{equation}

Since \(\phi\) is disentangling, there exists a PR box input state \(p_0\) that is mapped to a deterministic state.
At this input,
\begin{equation}
\overline A_X(p_0)=\pm1,
\qquad
M(p_0)=0
\end{equation}
for \(M\in\{A_X,A_Y,B_X,B_Y\}\).
Among the choices allowed by Lemma~\ref{lem:possible_local_marginals}, this is only possible for
\begin{equation}\label{eq:SM-vanishing-local-marginal}
\overline A_X=\alpha_X\textbf1\qquad(\alpha_X\in\{\pm1\}).
\end{equation}
Hence every pure input is mapped to a deterministic state.

By exactly the same argument as in the entangling case, we find that exactly two distinct input correlators occur, each exactly twice, and no constant factor occurs.
Moreover, both correlators occur either with opposite signs or with the same sign twice.

In particular, we have 
\begin{equation}\label{eq:SM-remained-outputs}
    \overline C_{Y,X},\overline C_{Y,Y}\in\{\pm C_{\vec W^{(X)}},\pm C_{\vec W^{(Y)}}\}.
\end{equation}
For $V\in\{X,Y\}$,
note that the relation
\(C_{X,V}=A_XB_V\) holds on \(\Sep\supseteq\phi(\ext(\NSstates))\).
Therefore, by combining Eq.~\eqref{eq:SM-output-correlator} with Eq.~\eqref{eq:SM-vanishing-local-marginal}, we have 
\begin{equation}\label{eq:SM-bob-output}
    \overline B_V=\beta_V C_{\vec W^{(V)}}\qquad(\beta_V\in\{\pm1\})
\end{equation}
on $\ext(\NSstates)$. 
By the affine property, Eq.~\eqref{eq:SM-bob-output} holds on the whole state space $\NSstates$.
Since $C_{\vec W^{(X)}}\notin \mathcal V_{\vec W'}$ 
for $\vec W'\ne\vec W^{(X)}$, 
Lemma~\ref{lem:SM-marginal-support}, Eq.~\eqref{eq:SM-remained-outputs} and Eq.~\eqref{eq:SM-bob-output} yield
\begin{equation}
    \overline C_{Y,X}=\alpha_Y\beta_X C_{\vec W^{(X)}}\qquad(\alpha_Y\in\{\pm1\}).
\end{equation}
By a similar argument to the derivation of Eq.~\eqref{eq:SM-bob-output}, we have
\begin{equation}
    \overline A_Y= \alpha_Y\textbf1,\qquad \overline C_{Y,Y}=\alpha_Y\beta_Y C_{\vec W^{(Y)}}.
\end{equation}

This is precisely the dynamical transformation given in
Eq.~\eqref{eq:SEP-PSP-A} in the main text.

The case in which distinct input correlators arise from different choices of Alice's measurement setting is treated analogously and yields the dynamical transformation of Eq.~\eqref{eq:SEP-PSP-B} in the main text.
\end{proof}

Entangling and disentangling PSP maps have been characterized above. For the remaining case, namely PSP maps that preserve the distinction between separable and entangled pure states, the following holds.

\begin{prop}[Completeness of entanglement-preserving PSP maps]\label{prop:reversibility}
A PSP map \(\phi\) on \(\NSstates\) is reversible if
\(\phi(\Sep)\subseteq\Sep\) and
\(\phi(\Ent)\subseteq\Ent\).
\end{prop}

\begin{proof}
Suppose that there exists
\(\vec W\in\{X,Y\}^2\)
such that
\(C_{\vec W}\neq\pm\overline C_{\vec Z}\)
for every
\(\vec Z\in\{X,Y\}^2\).

Choose a deterministic state \(p\) and a PR-box state \(p'\) whose correlation vectors differ only in the sign of \(C_{\vec W}\).
By Lemma~\ref{lem:heisenberg-correlator} in the main text,
\begin{equation}
\prod_{\vec Z\in\{X,Y\}^2}
\overline C_{\vec Z}(p)
=
\prod_{\vec Z\in\{X,Y\}^2}
\overline C_{\vec Z}(p').
\end{equation}

On the other hand, the assumptions
\(\phi(\Sep)\subseteq\Sep\) and
\(\phi(\Ent)\subseteq\Ent\)
imply that
\begin{align}
\prod_{\vec Z\in\{X,Y\}^2}
\overline C_{\vec Z}(p)
&=
\prod_{\vec Z\in\{X,Y\}^2}
C_{\vec Z}(\phi(p))
=
1,
\\
\prod_{\vec Z\in\{X,Y\}^2}
\overline C_{\vec Z}(p')
&=
\prod_{\vec Z\in\{X,Y\}^2}
C_{\vec Z}(\phi(p'))
=
-1.
\end{align}
This is a contradiction. Therefore,
\begin{equation}\label{eq:relabeling-correlator}
\{\pm\overline C_{\vec Z}\}_{\vec Z\in\{X,Y\}^2}
=
\{\pm C_{\vec W}\}_{\vec W\in\{X,Y\}^2}.
\end{equation}

Hence we may write
\begin{equation}
\overline C_{X,X}=\pm C_{\vec  W},
\qquad
\overline C_{X,Y}=\pm C_{\vec W'},
\qquad
(\vec W\neq\vec W').
\end{equation}
Lemma~\ref{lem:SM-marginal-support} then implies
\begin{equation}
\overline A_X
\in
\mathcal V_{\vec W}\cap\mathcal V_{\vec W'}.
\end{equation}

Among the possibilities listed in
Lemma~\ref{lem:possible_local_marginals},
the only ones consistent with the assumptions of \(\phi(\Sep)\subseteq \Sep\) and \(\phi(\Ent)\subseteq \Ent\) are
\begin{equation}
\overline A_X=\pm M,
\end{equation}
where 
\(M=A_{W_A}\) if $W_A=W'_A$ and 
\(M=B_{W_B}\) if $W_B=W'_B$.
We note that the case in which neither $W_A=W'_A$ nor $W_B=W'_B$ holds is excluded.

For the case of $W_A=W'_A$, 
we have \(W'_B=\check{W}_B(\neq W_B)\) and 
the same arguments for $\overline C_{Y,X}$ and $\overline C_{Y,Y}$ yield 
\begin{equation}
    \overline C_{Y,X}=\pm C_{\check{W}_A,W_B},
    \qquad
    \overline C_{Y,Y}=\pm C_{\check{W}_A,\check{W}_B},
\end{equation}
and 
\begin{equation}
    \overline A_Y=\pm A_{\check{W}_A},
    \qquad
    \overline B_X=\pm B_{W_B},
    \qquad 
    \overline B_Y=\pm B_{\check{W}_B}.
\end{equation}
Therefore, we have
\begin{equation}
\{\pm\overline A_X,\pm\overline A_Y,
\pm\overline B_X,\pm\overline B_Y\}
=
\{\pm A_X,\pm A_Y,\pm B_X,\pm B_Y\}.
\end{equation}
The same reasoning applies to the case of $W_B=W'_B$.
Therefore, for every pure state \(p\in\mathrm{ext}(\NSstates)\), there exists \(p'\in\mathrm{ext}(\NSstates)\) such that \(\phi(p')=p\).
Hence \(\phi(\NSstates)=\NSstates\), and this implies the reversibility of \(\phi\).
\end{proof}

Since every PSP map that is neither entangling nor disentangling preserves the distinction between separable and entangled pure states,
Theorem~\ref{thm:completeness} in the main text, Proposition~\ref{prop:EDcompleteness}, and Proposition~\ref{prop:reversibility} together imply the following.

\begin{thm}[Completeness of all PSP maps]\label{thm:PSPcompleteness}
Every PSP map on \(\NSstates\) is one of the \(3032\) maps presented in End Matter.
\end{thm}

\end{document}